\documentclass[runningheads]{llncs}
\usepackage{graphicx}

\usepackage{indentfirst, verbatim, dirtytalk,lineno, color}
\usepackage[utf8]{inputenc}
\usepackage[english]{babel}
\usepackage{amsmath,amssymb}
\usepackage{algorithmic}

\usepackage{caption}
\usepackage{algorithm}
\usepackage{algorithmic}
\DeclareCaptionLabelFormat{noname}{#2}
\definecolor{lavender}{rgb}{0.75, 0.58, 0.89}

\newlength\myindent
\newcommand\bindent[1][\myindent]{%
  \begingroup
  \setlength{\itemindent}{#1}
  \addtolength{\algorithmicindent}{#1}
}
\newcommand\eindent{\endgroup}

\begin{document}
\title{Analysis of LFT2}

\author{Geunwoo Lim,
Yujin Kwon,
Yongdae Kim}
\authorrunning{G. Lim et al.}

\institute{ Korea Advanced Institute of Science and Technology, Republic of Korea \and
\email{\{woo4303,dbwls8724,yongdaek\}@kaist.ac.kr}}
\maketitle              

\begin{abstract}
For a decentralized and transparent society, blockchain technology has been developed. 
Along with this, quite a few consensus algorithms that are one of core technologies in blockchain have been proposed. 
Among them, we analyze a consensus algorithm called LFT2, which is used by a blockchain system, ICON. 
We first formulate the LFT2 consensus algorithm and then analyze safety and liveness, which can be considered as the most important properties in distributed consensus system. 
We prove that LFT2 satisfies safety and liveness, where a certain assumption is required to prove liveness.
In addition, we compare LFT2 with two similar consensus algorithms, and from the comparison, we show that a trade-off exist among the three consensus algorithms.
Finally, we simulate LFT2 to measure a liveness quality.
\end{abstract}

\section{Introduction}

For a decentralized and transparent society, blockchain technology has been developed. 
The first blockchain system, Bitcoin~\cite{Nakamoto_bitcoin:a}, suggested a peer-to-peer electronic cash system using a proof of work (PoW) consensus algorithm.
In this system, nodes in a peer-to-peer network manage a distributed ledger called a blockchain in which transactions are stored. 
Each node writes latest transactions on a new block, which may be a part of blockchain, and then they propagate their block to other nodes. 
Next, each node determines whether to agree (or vote) on the received block, and the block is connected to the existing blockchain when enough votes are collected.
After that, a new round starts, and the above process is repeated. 
This process is conducted according to a consensus protocol.
Note that only transactions recorded on the blockchain are regarded as valid. 
Therefore, to ensure the security, it's important for the nodes to agree on the same ledger. 
In fact, due to a block propagation delay and the existence of attackers, nodes can have different views on blockchain.
In this case, nodes should resolve this, or the protocol should ensure that this never occurs.
If not, it allows an attacker to make an invalid transaction such as a double-spending transaction, and this undermines the system security significantly.
To resolve the above case, each node shouldn't vote for conflicting blocks.
Here, conflicting blocks have two types: 1) The first type is a block including conflicting transactions such as double spending transactions or invalid transactions. 
2) When two or more different blocks are committed in the same round, these blocks belong to the second type. 
The way to resolve conflicting blocks depends on each consensus protocol.

\textit{Safety} and \textit{liveness} can be considered as the most important properties of a consensus algorithm. 
A consensus algorithm should satisfy safety, which means the consensus algorithm doesn't commit conflicting blocks. 
Also, a consensus algorithm should satisfy liveness, which means the algorithm eventually extends a blockchain by adding a block. 
However, by FLP impossibility~\cite{fischer1982impossibility}, these cannot be satisfied in asynchronous network situation at the same time. 
Thus, each consensus algorithm has to choose which property to sacrifice. 
For example, the Nakamoto consensus sacrifices the safety property (liveness over safety) while a BFT-based consensus sacrifices the liveness property (safety over liveness).

Currently, many consensus algorithms including PoW and Byzantine Fault Tolerance (BFT) based consensus algorithms exist. 
In this paper, we focus on BFT-based consensus algorithms, and specifically, we analyze the LFT2~\cite{LFT2} consensus algorithm. 
We first model the LFT2 consensus algorithm and formalize it using a state machine. 
Then we prove that LFT2 satisfies safety and liveness properties in certain assumptions. 
Note that according to FLP impossibility, we cannot prove LFT2 satisfies both properties without any assumption. 
In addition, we define a metric called $\gamma-function,$ which can represent a liveness quality (i.e., a specific rate of creating a new block).  
We also compare the LFT2 consensus algorithm with two other BFT based consensus algorithms, PBFT~\cite{castro1999practical} and Hotstuff~\cite{yin2018hotstuff}. 
From this comparison, we find out trade-offs among these consensus algorithms.
Finally, we simulate LFT2 to measure a liveness quality using our metric, $\gamma-function.$

In summary, this paper makes the following contributions:

\begin{itemize}
    \item We formalize the LFT2 consensus algorithm.
    \item We prove the LFT2 consensus algorithm satisfies liveness and safety under certain assumptions.
    \item We simulate LFT2 to measure a liveness quality.
\end{itemize}

\section{Background}

\subsection{Practical Byzantine Fault Tolerance}
Practical Byzantine Fault Tolerance (PBFT)~\cite{castro1999practical} is a BFT-based consensus algorithm. 
Basically, the paper~\cite{castro1999practical} proposed it to prevent byzantine failures in a replication system. 
To put it simply, PBFT is a practical algorithm for consensus of a distribution system where at most $f$ byzantine nodes out of $3f+1$ nodes can exist in an asynchronous system.
PBFT is represented in Figure~\ref{PBFT_image}. 

Since PBFT is designed for a distributed replication system, it needs request and reply phases. 
However, because both phases aren't needed for consensus, PBFT would have only pre-prepare, prepare, and commit phases for consensus in a blockchain. 
Each round has a designated leader, and a leader proposes a new block and sends it to all other nodes at the start of round, where the new block would be what nodes should vote for. 
When the leader proposes a new block, pre-prepare phase starts.
After sending this process, prepare phase starts. In prepare phase, all nodes receive the block and check it. 
If the block is valid, then each node has to broadcast \say{Prepare} message to other nodes. 
While broadcasting the message, a node can receive Prepare message from other nodes. 
If a node receives $2f+1$ Prepare message, the node will be in \say{prepared} state. 
In commit phase, if a node is prepared state, each node broadcasts \say{Commit} message. 
After broadcasting, each node will receive Commit message, and if the number of Commit messages gotten from different nodes is greater or equal to $2f+1$, then the node accepts the block proposed by the leader. 
In this case, we state that the block is \textit{committed}.

\begin{figure}[h]
    \centering
    \includegraphics[width=\columnwidth]{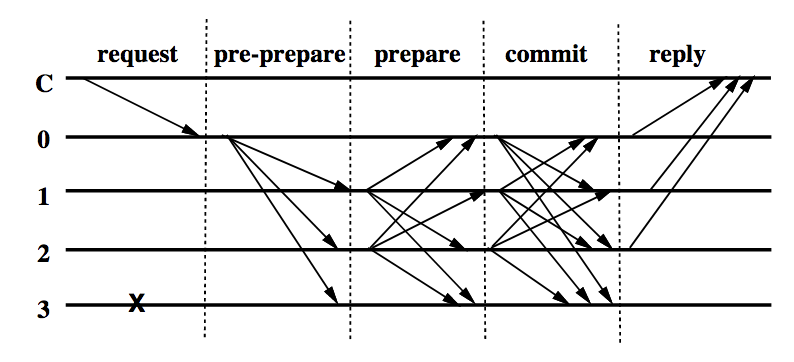}
    \caption[Practical Byzantine Fault Tolerance consensus description]{PBFT algorithm. This figure is from~\cite{castro1999practical}.
    } \label{PBFT_image}
\end{figure}

In addition, there is a view-change step in PBFT when a leader fails to send a new block to other nodes. 
If view-change occurs, the next leader sends another new block to other nodes. 
Because this step is a bit complicated, we don't describe this here. 
For more details, please refer the paper~\cite{castro1999practical}.

\subsection{Hotstuff}

Hotstuff is a leader-based BFT replication protocol published in 2018 by M.Yin et al~\cite{yin2018hotstuff}. 
A big difference between this protocol and PBFT is leader dependency. 
In PBFT, a leader is crucial only in pre-prepare phase, 
but in Hotstuff, the leader is crucial in every phase. 
This is because, in Hotstuff, each node sends a message to only a leader, and the leader should propagate the message to other nodes in every phase. 
This is a major difference between PBFT and Hotstuff, and this is why we refer Hotstuff to as a leader-based BFT protocol. 
Since each node doesn't directly broadcast a vote message to other nodes, the leader node can make some malicious actions. 
To prevent this, Hotstuff uses quorum certificate (QC), which is needed for proving that the leader receives $2f+1$ correct vote messages.

The basic algorithm of Hotstuff is represented in Figure \ref{Hotstuff_image}. 
Similar to PBFT, Hotstuff can successfully make nodes achieve consensus when at most $f$ number of faulty nodes exist in the system where the total number of nodes is $3f+1$. 
In prepare phase, each node sends a new-view message to a leader, and the leader receives $2f+1$ new-view messages. 
Then the leader sends prepare message to all nodes in the system, and each node receives the message. 
In pre-commit phase, each node checks the received prepare message, and if the message is valid, then it sends the prepare vote message to the leader. If the leader receives $2f+1$ prepare vote messages, then the leader sends pre-commit message to each node. 
In commit phase, each node validates the pre-commit message and if the message is proper, then it sends pre-commit vote message to the leader. 
If the leader receives $2f+1$ pre-commit vote messages, then the leader sends commit message to each node. 
In decide phase, each node checks the commit message gotten from the leader, and if the message is proper, then a node sends commit vote message to the leader. 
Similar to the previous phases, if the leader receives $2f+1$ commit vote messages, the leader makes a decide message and sends to each node. 
In all the above steps, a leader can be a byzantine node for some reasons. 
Each node can send a new-view message to the next leader when the node judges the current leader node is byzantine. 
If the next leader receives $2f+1$ new-view messages, then the new leader starts consensus by repeating the above process. 
In fact, this new-view message process is similar to a view-change process in PBFT.

\begin{figure}[h]
    \centerline{\includegraphics[width=\columnwidth]{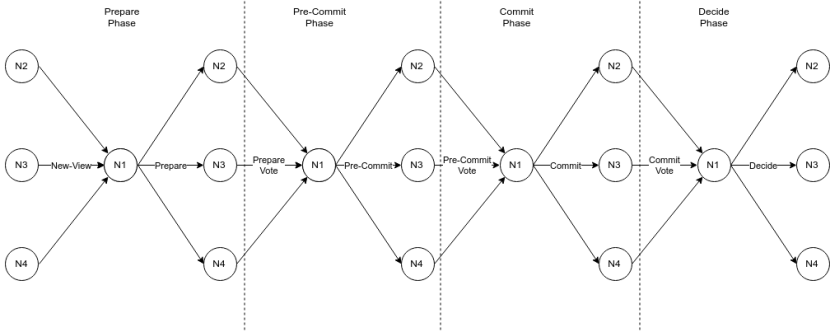}}
    \caption[Hotstuff consensus description]{Hotstuff algorithm. This figure is from~\cite{hotstuffimage}
    } \label{Hotstuff_image}
\end{figure}

As shown in Figure~\ref{Hotstuff_image}, we can see that each round is symmetric.
Thus, considering this characteristic, the Hotstuff whitepaper~\cite{yin2018hotstuff} suggests an advanced Hotstuff called Chained Hotstuff, which has higher scalability. 
Chained Hotstuff can be simply explained as a pipelining version of Hotstuff. 
This is described in Figure~\ref{Chained_Hotstuff_image}. 
In each round, a leader node makes a message linked with the previous block, and sends it to other nodes. 
Each node receives the message from the leader and then verifies the message. 
The received messages include four parts, and each node has to check the whole four parts of the messages. 
In summary, Chained Hotstuff uses pipelining in each round, and as a result, this obtains scalability without losing the security since consensus steps aren't reduced. 

\begin{figure}[h]
    \centering
    \includegraphics[width=\columnwidth]{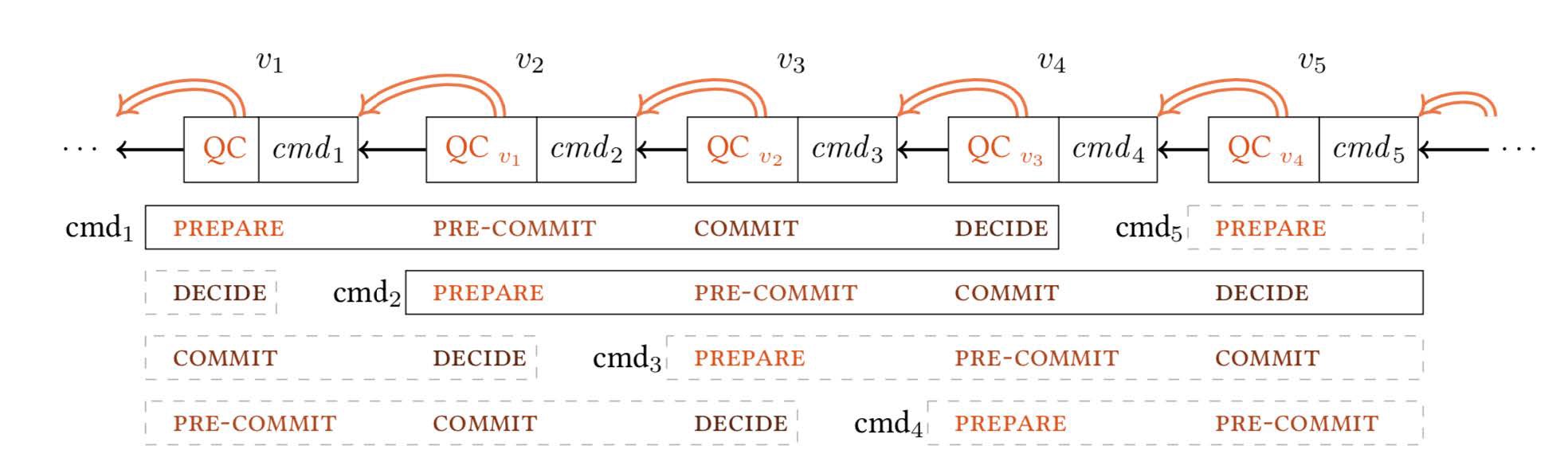}
    \caption[Chained Hotstuff Pipelining]{Chained Hotstuff does pipelining using a symmetric structure. 
    This figure is from~\cite{yin2018hotstuff}
    } \label{Chained_Hotstuff_image}
\end{figure}

\subsection{LFT2}

The basic algorithm of LFT2 is similar to PBFT. 
A leader makes a new block and broadcasts it to other nodes. 
Each node knows a leader of each round. 
Thus, each node can check that the received block is proposed by a proper (or valid) leader. 
If the received block doesn't have any error, then each node makes and broadcasts a vote message. 
If a node gets enough vote messages, then the block would be a candidate block and the previous candidate block becomes a committed block. 
Figure~\ref{LFT2_image} represents LFT2.

\begin{figure}[h]
    \centering
    \includegraphics[width=\columnwidth]{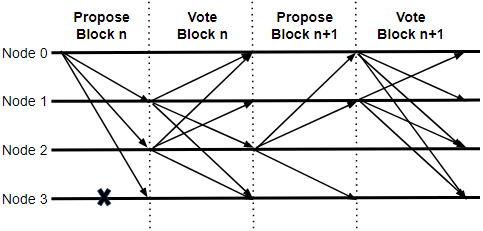}
    \caption[LFT2 consensus description]{LFT2 algorithm. Block $n$ will be a candidate block after voting for block $n$, and it's finally committed after voting for block $n+1$. 
    Obviously, after Vote Block $n+1$ phase, block $n+1$ will be a candidate block.
    } \label{LFT2_image}
\end{figure}

To explain a specific algorithm of LFT, assumptions and rules of LFT2 are required. 
In a LFT2 consensus algorithm, there are two types of nodes, honest nodes and byzantine nodes. 
These nodes do following actions.

\begin{itemize}
    \item General node
    \begin{itemize}
        \item If a node receives a message, it sends the message to neighbor nodes at a specific time (gossip communication).
        \item A leader selection algorithm exists and each node knows the order.
        \item Each node has a local timer.
        \item Each node already knows a cipher suite, and every message includes digital signature.
    \end{itemize}
    \item Byzantine node
    \begin{itemize}
        \item A byzantine node can delay or may not send a message.
        \item A byzantine node can send different messages to different nodes.
        \item A byzantine node cannot generate other node’s digital signature.
    \end{itemize}
\end{itemize}

In LFT2, there exist two steps, propose and vote. 
In propose step, a leader proposes a new block and broadcasts to other nodes. 
In vote step, each node checks the received block and sends a vote message to other nodes. 
The basic rule of LFT2 is as follows. 

\begin{itemize}
    \item ProposeTimer
    \begin{itemize}
        \item If propose step starts, ProposeTimer works.
        \item If ProposeTimeout occurs, a failure vote progresses.
    \end{itemize}
    \item VoteTimer
    \begin{itemize}
        \item In vote step, if enough vote messages arrive but consensus isn't completed, Votetimer works.
        \item If consensus completes within a fixed time, VoteTimer stops.
        \item If consensus doesn't complete within a fixed time, VoteTimeout occurs.
    \end{itemize}
    \item Propose
    \begin{itemize}
        \item A leader makes only one block in one round.
    \end{itemize}
    \item Vote
    \begin{itemize}
        \item A validator receives a new block from the leader.
        \item Validators check the block information (if it's proposed by a proper leader, has the correct previous hash, and is connected with the candidate block, etc).
        \item If the block is valid, validators send a vote message to other nodes.
    \end{itemize}

    \item Candidate
    \begin{itemize}
        \item If a node receives enough votes for a block of a higher round or a block with a higher height than the candidate block that the node views, the node changes this block to a candidate block.
    \end{itemize}
    \item Commit
    \begin{itemize}
        \item A node commits the previous candidate block when replacing the current candidate block with a new candidate block.
    \end{itemize}
\end{itemize}

We define states and transitions for each node based on above rules in Section~\ref{chap:formalization}.

\section{System Formalization}
\label{chap:formalization}

\subsection{Definitions}

In this section, we formalize the LFT2 protocol. 
This formalization would help to understand the LFT2 consensus protocol and prove \textit{safety} and \textit{liveness.} 
Before the formalization, we first define variables. 
For the network formalization, $\mathcal{M}$ represents a message, which transmits from a node to other nodes, and $N$ indicates the number of nodes in the system. 
For the node formalization, we define two variables: $in$ and $out$. 
The parameter $in$ indicates a message set including valid new information, and $out$ indicates a message set, which a node has to send to other nodes. 
The parameter $h$ indicates a block height, $T$ and $V$ are a set of collected transactions and a set of votes for a candidate block, respectively. 
Our model follows almost the paper~\cite{castro1999correctness}. 

\subsection{The Multicast Channel Automaton}

\begin{algorithm}[h]
\begin{algorithmic}
\caption{\textbf{State}}
\STATE $net \subseteq \mathcal{M} \times 2^N$, initially \{\}
\end{algorithmic}
\end{algorithm}

\begin{algorithm}[h]
\begin{algorithmic}
\caption{\textbf{Transitions}}
\STATE SEND$(m,X)${
\bindent
    \STATE Eff: $net:=net\cup \{(m,X)\}$ 
\eindent}
\STATE RECEIVE$(m)_x${
\bindent
    \STATE Pre: $\exists (m,X)\in net : (x\in X$
    \STATE Eff: $net:=net-\{(m,X)\}\cup \{(m,X-{x})\})$
\eindent}
\STATE MISBEHAVE$(m,X,X')${
\bindent
    \STATE Pre: $(m,X)\in net $
    \STATE Eff: $net:=net-\{(m,X)\}\cup \{(m,X')\})$
\eindent}
\end{algorithmic}
\end{algorithm}

In this section, we model a network state and transitions of a blockchain system. 
The variable $net$ represents a state of network, and $2^N$ is a vector that represents whether each node received or not received a message. 
SEND$(m,X)$ means sending message $m$ to node set $X$, so state $net$ should include $(m,X)$ after SEND$(m,X).$
RECEIVE$(m)_x$ means that node $x$ receives message $m,$ so we should update $net$ state as follows. 
Because node $x$ received the message according to RECEIVE$(m)_x,$ we remove $(m,X)$ from $net$ and add $(m,X-x)$ to $net$.
MISBEHAVE$(m,X,X')$ is similar with a RECEIVE process. 
The difference is because of misbehavior, the receiving node set will be $X'$ instead of $X.$

\subsection{The Replica Automaton}

\begin{algorithm}[h]
\begin{algorithmic}
\caption{\textbf{Auxiliary functions}}
\STATE $proposer(h)$ $\rightarrow $ block height (input), proposer node ID (output)
\STATE $check\_hash(B)$ $\rightarrow $ block info (input), boolean (output)
\STATE $node\_candidate\_block(i)$ $\rightarrow $ node ID (input), the input node's candidate block info (output)
\STATE $block\_candidate\_block(B)$ $\rightarrow $ block info (input), the input block's candidate block info (output)
\STATE $block\_height(B)$ $\rightarrow $ block info (input), the input block's height (output)
\STATE $node\_height(i)$ $\rightarrow $ node info (input), the input node's candidate block height (output)
\STATE $same\_vote(in)$ $\rightarrow $ the largest set of the same votes in $in$
\end{algorithmic}
\end{algorithm}

These auxiliary functions use for representing below transitions. 
Each function has own input and output values.
For example, $proposer(h)$ function has block height $h$ an an input and outputs the proposer node ID for $h.$ 
This function can use to check if the proposer (leader) is valid at block height $h$. 
The function $check\_hash(B)$ has block information $B$ as an input and outputs boolean. 
This function checks the correctness of hash value of block $B$ and returns if the value is true or false. 
Lastly, $same\_vote(in)$ means a maximal set of same votes in $in$. 
This is used for checking whether enough votes are collected.

\subsection{The Replica Automaton}

\begin{algorithm}[h]
\begin{algorithmic}
\caption{\textbf{Input transitions}}
\label{alg:in}
\STATE RECEIVE($\langle $NEW-BLOCK$,B(h,T,V)\rangle_{\sigma_j})_i${
\bindent
    \STATE Pre: $j=proposer(h)\land height_i=h-1\land check\_hash(B)\land state_i=ready\land$
    \bindent[1.45cm] 
    \STATE $node\_candidate\_block(i)=block\_candidate\_block(B)$ \eindent
    \STATE Eff: $in_i:=in_i\cup \{\langle $NEW-BLOCK$,B(h,T,V)\rangle\}$
    \bindent[1.4cm]
    \STATE $state_i=process$
    \STATE $SEND(\{\langle $VOTE$,V(N,B',B)\rangle\})$\eindent
\eindent}
\STATE RECEIVE($\langle $VOTE$,V(N,B',B)\rangle_{\sigma_j})_i${
\bindent
    \STATE Pre: $(\langle $NEW-BLOCK$,B(h,T,V)\rangle\in in_i \land \langle $VOTE$,V(N,B',B)\rangle_{\sigma_j}\notin in_i)\;or$
    \bindent[1.45cm] 
    \STATE $(block\_height(B)>node\_height(i)\land $VOTE$,V(N,B',B)\rangle_{\sigma_j}\notin in_i)$ \eindent
    \STATE Eff: $in_i:=in_i\cup\{\langle $VOTE$,V(N,B',B)\rangle_{\sigma_j}\}$
\eindent}
\STATE RECEIVE($\langle $TIMEOUT$,BLANK)_i${
\bindent
    \STATE Pre: $\langle $NEW-BLOCK$,B(h,T,V)\rangle\notin in_i \land \langle $TIMEOUT$,BLANK\rangle_{\sigma_j}\notin in_i$
    \STATE Eff: $in_i:=in_i\cup \{\langle $TIMEOUT$,BLANK\rangle_{\sigma_j}\rangle\}$
\eindent}
\end{algorithmic}
\end{algorithm}

\begin{algorithm}[h]
\begin{algorithmic}
\caption{\textbf{Output Transition}}
\label{alg:out}
\STATE SEND($m$)$_i${
\bindent
    \STATE Pre: $m\in out_i$
    \STATE Eff: $out_i:=out_i-\{m\}$
\eindent}
\end{algorithmic}
\end{algorithm}

The above transitions represent state transitions when there exist new inputs or outputs. 
Input transitions occur when a message is received from other nodes, and output transitions occur when a message should be sent to other nodes. 
In Algorithms~\ref{alg:in}, \ref{alg:out}, and \ref{alg:internal}, Pre indicates conditions required for the corresponding transition, and Eff indicates the result of the transitions. 
For example, RECEIVE($\langle $NEW-BLOCK$,B(h,T,V)\rangle_{\sigma_j})_i$ means that node $i$ receives message $\langle $NEW-BLOCK$,B(h,T,V)\rangle$ from node $j,$ and it should satisfy conditions in Pre (see Algorithm~\ref{alg:in}). 
The result of the transition is represented in Eff. 
RECEIVE($\langle $NEW-BLOCK$,B(h,T,V)\rangle_{\sigma_j})_i$ occurs when a NEW-BLOCK message is arrived.
RECEIVE($\langle $VOTE$,V(N,B',B)\rangle_{\sigma_j})_i$ occurs when a VOTE message is arrived, and RECEIVE($\langle $TIMEOUT$,BLANK)_i$ occurs when a TIMEOUT message is arrived.
As another example, we describe SEND$(m)_i$ represented in Algorithm~\ref{alg:out}. 
This means that node $i$ sends message $m$ to other nodes, so $m$ should be in $out_i$. 
After $SEND(m)_i$ occurs, $m$ should be removed from $out_i$. 

\begin{algorithm}[h]
\begin{algorithmic}
\caption{\textbf{Internal transitions}}
\label{alg:internal}
\STATE SEND\_BLOCK($h,T,V$)$_i${
\bindent
    \STATE Pre: $proposer(h)=i\land node\_height(i)=h-1 \land state_i=ready$
    \STATE Eff: $out_i:=out_i\cup \{\langle $NEW-BLOCK$,B(h,T,V)\rangle\}$
    \bindent[1.4cm]
    \STATE $SEND(\{\langle $NEW-BLOCK$,B(h,T,V)\rangle\})$\eindent
\eindent}
\STATE SEND\_TIMEOUT$_i${
\bindent
    \STATE Pre: $\{\langle $NEW-BLOCK$,B(h,T,V)\rangle\}\notin in_i$
    \STATE Eff: $out_i:=out_i\cup \{\langle $TIMEOUT$,BLANK\rangle\}$
    \bindent[1.4cm]
    \STATE $SEND(\{\langle $TIMEOUT$,BLANK\rangle\})$\eindent
\eindent}
\STATE COMMIT$_i${
\bindent
    \STATE Pre: $|same\_vote(in_i)|>=2f+1$
    \STATE Eff: $for\;V(N,B',B) \in same\_vote(in_i),$
    \bindent[1.4cm]
    \STATE $commit(B')$
    \STATE $candidate(B)$
    \STATE $node\_height(i):=block\_height(B')$
    \STATE $state_i:=ready$\eindent
\eindent}
\STATE VOTE\_FAIL$_i${
\bindent
    \STATE Pre: $\langle $NEW-BLOCK$,B(h,T,V)\rangle\in in_i\land |same\_vote(in_i)|<2f+1 $
    \STATE Eff: $state_i:=ready$
\eindent}
\STATE VOTE\_TIMEOUT$_i${
\bindent
    \STATE Pre: $\langle $NEW-BLOCK$,B(h,T,V)\rangle\notin in\land |same\_vote(in_i)|>=2f+1$
    \bindent[1.45cm]
    \STATE $(for \;m \in same\_vote(in_i), type(m)=TIMEOUT)$\eindent
    \STATE Eff: $state_i=ready$
\eindent}
\end{algorithmic}
\end{algorithm}

Algorithm~\ref{alg:internal} represents state transitions that occur internally in a node. 
SEND\_BLOCK($h,T,V$)$_i$ occurs when node $i$ is a proposer so the node has to make a new block and send it to other nodes.
SEND\_TIMEOUT$_i$ occurs when a node has to receive NEW-BLOCK message but no NEW-BLOCK message arrived until $propose\_timeout$ occurs. 
COMMIT$_i$ occurs when enough VOTE messages are received, where votes are for committing a new block. 
VOTE\_FAIL$_i$ occurs when $vote\_timeout$ occurs, which means enough VOTE messages aren't received (i.e., the number of elements in $same\_vote(in)$ does not satisfy $2f+1$). 
In this case, consensus is failed. 
Lastly, VOTE\_TIMEOUT$_i$ occurs when enough TIMEOUT messages arrive, which means the leader node doesn't send a NEW-BLOCK message.

\subsection{Consensus Process}

Next, we formalize the LFT2 consensus process by using transitions. 
First, a leader should send a new block to other nodes. 
Other nodes wait for a new block but if it doesn't arrive until $propose\_timeout$ occurs, then send TIMEOUT message and vote about the message. 
On the other hand, a new block arrives in time, then the node sends VOTE message to other nodes. 
In vote phase, if $vote\_timeout$ occurs, then VOTE\_TIMEOUT will execute. 
If $vote\_failure$ occurs, then VOTE\_FAIL will execute, 
and if both $vote\_timeout$ and $vote\_failure$ don't occur and enough votes for the new block message arrive, then COMMIT will execute.
In this case, the round successfully ends. 

\begin{algorithm}[h]
\begin{algorithmic}
\caption{\textbf{Consensus process}}
\STATE SEND\_BLOCK($h,T,V$) (Leader Only){
\IF{$propose\_timeout$ occurs}
    \bindent[1.45cm] 
    \STATE SEND\_TIMEOUT
    \STATE RECEIVE($\langle $TIMEOUT$,BLANK$)
    \STATE VOTE\_TIMEOUT\eindent
\ELSE
    \bindent[1.45cm]
    \STATE RECEIVE($\langle $NEW-BLOCK$,B(h,T,V)\rangle_{\sigma_j}$)
    \STATE SEND($\langle $VOTE$,V(N,B',B)\rangle$)\eindent\ENDIF
\IF{$vote\_timeout$ occurs}
    \bindent[1.45cm] 
    \STATE RECEIVE($\langle $TIMEOUT$,BLANK$)
    \STATE VOTE\_TIMEOUT\eindent
\ELSIF{$vote\_failure$ occurs}
    \bindent[1.45cm] 
    \STATE RECEIVE($\langle $TIMEOUT$,BLANK$)
    \STATE VOTE\_FAIL\eindent
\ELSE
    \bindent[1.45cm] 
    \STATE RECEIVE($\langle $VOTE$,V(N,B',B)\rangle_{\sigma_j}$)
    \STATE COMMIT\eindent\ENDIF}
\end{algorithmic}
\end{algorithm}

\section{Safety and Liveness}

In this chapter, we first define \textit{safety} and \textit{liveness} and prove that LFT2 satisfies safety and liveness under certain conditions. 
In addition, to measure liveness quality of LFT2, we suggest a metric, $\gamma-function,$ which represents an average rate of generating committed blocks. 

\subsection{Safety}
Before proving that the LFT2 satisfies the safety property, we define safety below. 

\begin{definition}{(Safety)}
We state that a blockchain system satisfies \textbf{safety} when a conflicting block never commits.
\end{definition}

The intuitive meaning of satisfying a safety property is conflicting blocks will never commit.
In other words, safety asserts that nothing bad thing happens, where bad things mean that conflicting blocks are committed. 
Conflicting blocks include two types: 
The first type is a block including conflicting transactions such as double-spending transactions or invalid transactions. 
The second type indicates blocks when they are committed at the same height.
In Lemmas~\ref{safetylem1} and \ref{safetylem2}, we prove that the second and first types of conflicting blocks will not commit in LFT2, respectively.

\begin{lemma}\label{safetylem1}
In LFT2, two different blocks cannot be committed at the same height when at most $f$ byzantine nodes exist in a system where there are $3f+1$ nodes.
\end{lemma}

\begin{proof}
Let's assume two different blocks $B_1$ and $B_2$ committed at height $h$. 
This means that there are at least $2f+1$ nodes who committed block $B_1$ and another at least $2f+1$ nodes who committed block $B_2$. Therefore, this means at least $4f+2$ votes provided by nodes exist.
Because at most $f$ nodes can vote two times simultaneously, at most $4f+1$ votes can exist, which implies that two different blocks cannot be committed at the same height.
\end{proof}

\begin{lemma}\label{safetylem2}
In LFT2, two conflict transactions cannot be committed when at most $f$ byzantine nodes exist in a system where there are $3f+1$ nodes.
\end{lemma}

\begin{proof}
By Lemma~\ref{safetylem1}, in LFT, a fork (i.e., committing different blocks at the same height) cannot be happen. 
Without a fork, to commit conflict transaction, there should be at least $2f+1$ malicious nodes. 
However, because at most $f$ byzantine nodes can exist, conflict transactions cannot be committed.
\end{proof}

From both lemmas, we show that LFT2 satisfies the safety property.

\begin{theorem}
LFT2 satisfies a safety property.
\end{theorem}

\begin{proof}
In LFT2, there is no way to invalidate an already committed block. Therefore, if there is no conflicting block committed, the safety property would satisfy. 
By Lemmas~\ref{safetylem1} and \ref{safetylem2}, there is no conflicting block committed so LFT2 satisfies the safety property.
\end{proof}

\subsection{Liveness}

Next, we analyze liveness of LFT2. 
We first define liveness below. 

\begin{definition}{(Liveness)}
We state that a blockchain system satisfies \textbf{liveness} when a new block is committed without stuck.
\end{definition}

Intuitively, satisfying a liveness property is that a new block will be committed continuously without any stuck.
In other words, liveness asserts that something a good thing eventually happens, where a good thing implies committing a block.

Now, using some assumptions, we prove liveness of LFT2.
Note that we cannot prove liveness without any network assumptions according to FLP impossibility~\cite{fischer1982impossibility}.

\begin{theorem}\label{thm1}
We assumes that every non-byzantine node $i$ is in $state_i=ready$ at some specific time, where $state_i=ready$ implies that node $i$ is ready to enter a new round. 
We also define $d$ as a time difference between when the first node gets the next $ready$ and when the last node gets the next $ready.$ 
Then if both $d+\Delta <ProposeTimeout<\infty$ and $2\Delta < VoteTimeout<\infty$ are satisfied, consensus would complete. 
Furthermore, if the leader is a non-byzantine node, a new block will be committed eventually.
\end{theorem}

\begin{proof}
Assume that both nodes $A$ and $B$ are non-byzantine, and nodes $A$ and $B$ are the first and last nodes gotten state $ready,$ respectively. 
That is, node $A$ first got $state_A=ready,$ and node $B$ lastly got $state_B=ready.$
Let $t$ be the time when node $A$'s state became $ready$, i.e. the time when node $B$'s state became $ready$ would be $t+d$ according to the definition of $d.$ 
At time $t$, a proposer timer of node $A$ works until $t+ProposeTimeout$. 
If node $A$ didn't receive the NEW-BLOCK message until $t+ProposeTimeout$, then the node $A$ will send TIME-OUT message. 
To prove the theorem, let's consider the worst case. 
The worst case is when the new round leader is node $B$, which means that a new round leader is the latest node whose state became $ready$. 
To complete consensus, NEW-BLOCK message should arrive at node $A$ until $t+ProposeTimeout$. 
However, because we assumed that node $B$ is the leader and the time when $state_B=ready$ is $t+d$, the maximum time when node $A$ receives NEW-BLOCK will be $t+d+\Delta$. 
As a result, to complete the consensus,
\begin{eqnarray*} t+d+\Delta <t+ProposeTimeout\end{eqnarray*} 
should be met.

After sending a NEW-BLOCK message, if a vote timer of node $B$ works until $t+d+VoteTimeout$. 
If node $B$ doesn't receive enough VOTE messages until $t+d+VoteTimeout$, then node $B$ would think that consensus is failed and it'd go to next step. 
Thus, to complete the consensus, enough VOTE messages should arrive at node $B$ until $t+d+VoteTimeout$. 
The worst case in this situation is that the VOTE message sender receives NEW-BLOCK message at $t+d+\Delta$ and then sends the VOTE message but the VOTE message arrives at $t+d+\Delta+\Delta$ to $B$. 
As a result, to complete the consensus,\begin{eqnarray*} t+d+2\Delta <t+d+VoteTimeout\end{eqnarray*} should be satisfied. 

Obviously, if there is no $ProposeTimeout$ and a leader is malicious, then the leader may not send a message, so replicas never get NEW-BLOCK message. 
This means that consensus will be never completed. 
Similarly, if there is no $VoteTimout$ and a leader is malicious, then the leader can generate a conflicting NEW-BLOCK message, so replicas never receive enough VOTE messages. 
This also makes consensus never complete. 
Thus, 
\begin{eqnarray*} 
ProposeTimeout < \infty \\VoteTimeout < \infty 
\end{eqnarray*} 
should be satisfied. 
Finally, if the following equation
\begin{eqnarray*} 
d+\Delta < ProposeTimeout < \infty \\ 2\Delta < VoteTimeout < \infty 
\end{eqnarray*}
is satisfied, then the consensus is completed.

After the consensus is completed, we assume that a new block is not committed, which means that timeout occurs at some point. 
Also, we assume that $ProposeTimeout$ occurs. 
Since the number of byzantine nodes is less or equal to $f$ and $d+\Delta < ProposeTimeout< \infty$ is satisfied, 
the only case that $ProposeTimeout$ can occur is when the leader is byzantine, which is a contradiction. 

Next, let's assume $VoteTimeout$ occurs. 
Similarly, since the number of byzantine nodes is less or equal to $f$ and $2\Delta < VoteTimeout< \infty$ is satisfied, the only case that $VoteTimeout$ can occur is when the leader is malicious. 
The $VoteTimeout$ occurrence means that not enough VOTE messages arrive but no more than $f$ number of byzantine nodes exist. 
This means that conflicting NEW-BLOCK messages are propagated to each node. 
The only node who can make conflicting NEW-BLOCK messages is the leader.
This is because if one of other nodes makes conflicting NEW-BLOCK messages, the message would be filtered by each node using $proposer$ function. 
In conclusion, both timeout cases can occur only when the leader is byzantine, but this is a contradiction. 
Therefore, if the leader is a non-byzantine node, a new block would be committed.
\end{proof}

Moreover, we define $\gamma-function$ to measure how good LFT2 has  liveness.

\begin{definition}
Let $\nu$ be the total number of faulty nodes in the LFT2 system, we define $\gamma-function$ (i.e., $\gamma (\nu)$) as follows: 
\begin{eqnarray*} 
\gamma(\nu) = \lim\limits_{m \to \infty} {\mathbb{E}[X_{m,\nu}]},\ where\ X_{m,\nu}={committed\_block(m,\nu)\over m}
\end{eqnarray*} 
The term of $committed\_block(m,\nu)$ indicates the number of committed blocks during changing a leader $m$-times if the system has $\nu$ faulty nodes.
\end{definition}

Then we prove the minimum average rate of generating committed blocks is $\gamma (f)$ in LFT2.

\begin{lemma}\label{lemma1}
If $\nu \leq f$, $d+\Delta <ProposeTimeout<\infty$ and $2\Delta < VoteTimeout<\infty$ are satisfied at each round, then $\gamma (\nu)= \cfrac{n-\nu}{n}$. Here, $d$ indicates a time difference between when the first node gets the next $ready$ and when the last node gets the next $ready$.
\label{lem:gamma}
\end{lemma}

\begin{proof}
LFT2 uses a round-robin type of a leader selection algorithm. 
Therefore, each node would be a leader once during changing a leader $n$ times. 
By Theorem~\ref{thm1}, if a leader is a non-byzantine node, a new block would be committed by consensus. 
Thus, 
\begin{eqnarray*}\mathbb{E}[committed\_block(n,\nu)] = n-\nu\end{eqnarray*} 
is satisfied. 
By the division theorem, we can represent $m$ as follows:
\begin{eqnarray*} m=nq+r, \ where \ 0\leq r < n \end{eqnarray*} 
Using both results, we can represent $\mathbb{E}[committed\_block(m,\nu)]$ as follows:

\begin{align*} 
\mathbb{E}[committed\_block(m,\nu)] &=\mathbb{E}[committed\_block(nq+r,\nu)]\\ \Rightarrow q(n-\nu) &=\mathbb{E}[committed\_block(nq,\nu)] \\ &\leq\mathbb{E}[committed\_block(m,\nu)] \\&\leq \mathbb{E}[committed\_block(n(q+1),\nu)] \\&=(q+1)(n-\nu) \end{align*} 

\begin{eqnarray*}\therefore \ q(n-\nu)\leq \mathbb{E}[committed\_block(m,\nu)] \leq (q+1)(n-\nu)\end{eqnarray*} 

Therefore, we can represent $\mathbb{E}[X_{m,\nu}]$ as below: 

\begin{eqnarray*}\mathbb{E}[X_{m,\nu}] =\mathbb{E}\left[{committed\_block(nq+r,\nu)\over nq+r}\right]\end{eqnarray*} \\ \begin{eqnarray*} \therefore {q(n-\nu) \over nq+r}\leq \mathbb{E}[X_{m,\nu}]\leq {(q+1)(n-\nu) \over nq+r}\end{eqnarray*}
Thus, $\gamma (\nu)$ as follow: 
\begin{align*} \gamma (\nu) &=\lim\limits_{m \to \infty} {\mathbb{E}[X_{m,\nu}]}\\ &= \lim\limits_{q \to \infty} {q(n-\nu) \over nq+r}\\ &= \lim\limits_{q \to \infty} {q(n-\nu)+(n-\nu) \over nq+r}\\ &={n-\nu \over n}
\end{align*}
This completes the proof.
\end{proof}

\begin{theorem}\label{livenessmetric}
In LFT2, if $d+\Delta <ProposeTimeout<\infty$ and $2\Delta < VoteTimeout<\infty$ are satisfied at each round, then the following inequality is satisfied.
\begin{eqnarray*} \lim\limits_{m \to \infty} {\mathbb{E}[X_m]}\geq \gamma(f),\ where\ X_m={committed\_block(m)\over m}\end{eqnarray*}
Here, $d$ indicates a time difference between when the first node gets the next $ready$ and when the last node gets the next $ready$. 
Parameter $committed\_block(m)$ represents the number of committed blocks during changing a leader $m$ times when the system has at most $f$ faulty nodes.
\end{theorem}

\begin{proof}
Let $m_i$ be the number of rounds in which $i$ faulty nodes exist during changing a leader $m$ times. 
Obviously, $m=\sum\limits_{i=0}\limits^f m_i.$ 
Here, note that when a leader changes, one round passes.
Then we can represent $X_m$ as follow:

\begin{align*} 
X_m&=X'_{m,0}+X'_{m,1}+\dots+X'_{m,f-1}+X'_{m,f}\\&={m_0 \over m} X_{m_0,0}+{m_1 \over m} X_{m_1,1}+\dots +{m_{f-1} \over m} X_{m_{f-1},f-1} +{m_f \over m} X_{m_f,f}, \\&where \ X'_{m,\nu}={partially\_committed\_block(m,\nu)\over m} 
\end{align*}

The term $partially\_committed\_block(m,\nu)$ means the number of committed blocks when the system has $\nu$ faulty nodes during changing a leader $m$-times. 
Using the linearity of expectation, the following equation is satisfied:
\begin{eqnarray*} \mathbb{E}[X_m]= {m_0 \over m}\mathbb{E}[X_{m_0,0}]+ {m_1 \over m}\mathbb{E}[X_{m_1,1}] +\dots +{m_{f-1} \over m}\mathbb{E}[X_{m_{f-1},f-1}] +{m_f \over m}\mathbb{E}[X_{m_f,f}]\end{eqnarray*}
By Lemma~\ref{lemma1}, $$\gamma(0)\geq\gamma(1)\geq\dots\geq\gamma(f-1)\geq\gamma(f)$$ is satisfied. 
Therefore, the following inequality is hold:
\begin{align*} \lim\limits_{m \to \infty}\mathbb{E}[X_m]&= \lim\limits_{m \to \infty}\left( {m_0 \over m}\gamma(0)+ {m_1 \over m}\gamma(1) +\dots +{m_{f-1} \over m}\gamma(f-1) +{m_f \over m}\gamma(f)\right)\\&\geq \lim\limits_{m \to \infty}\left( {m_0 \over m}+ {m_1 \over m} +\dots +{m_{f-1} \over m} +{m_f \over m}\right)\gamma(f)\\ &=\gamma(f)\end{align*}
This completes the proof.
\end{proof}

\section{Comparative Analysis}

In this section, we compare LFT2 with other two consensus protocols. 
From this comparison, we study trade-offs among the three consensus algorithms.

\subsection{PBFT vs LFT2 vs Chained Hotstuff}

\smallskip\noindent\textbf{Scalability. } 
LFT2~\cite{lft2_white} is between PBFT~\cite{castro1999practical} and Hotstuff~\cite{yin2018hotstuff} logically.
When broadcasting a vote message, LFT2 is similar to PBFT. 
On the other hand, because it doesn't commit a new block in one round, we can say that LFT2 is similar to Chained Hotstuff. 
Without any byzantine nodes, the average number of phases required for committing a block is one, two, and three for Chained Hotstuff, LFT2, and PBFT, respectively. 
All three consensus algorithms need three steps during one block committing process, but the average number of phases required to commit a block depends on how they apply pipelining technique.
For this reason, in terms of scalability, Chained Hotstuff has the best performance while PBFT is the worst.

\smallskip\noindent\textbf{Network Bandwidth. } 
PBFT has three phases in one round; a leader first broadcasts a message to other nodes, the nodes vote the message received from the leader in the second phase, and the nodes send the pre-commit message based on the received votes in the last phase. 
Using big-O notation, the first, second, and third phases would have the network bandwidth complexity of $O(N),\ O(N^2),$ and $O(N^2)$, respectively.
Here, $N$ means the number of nodes in the consensus system.
LFT2 has two phases where a leader first broadcasts a message to other nodes and then the nodes vote the message that the leader sent in the second phase. 
Using big-O notation, the first and second phases have the network bandwidth complexity of $O(N)$ and $O(N^2)$, respectively. 
Chained Hotstuff has only one phase where a leader first broadcasts a message to other nodes and then, similar to a voting process, each node responses to the leader.
Using big-O notation, each phase would have the network bandwidth complexity of $O(N)$.
Considering the above, we find out that Chained Hotstuff has the best network bandwidth complexity, compared to PBFT and LFT2. 
In addition, even though PBFT and LFT2 have the same network bandwidth complexity in terms of Big-O notation, 
we can state that LFT2 has better network bandwidth complexity than PBFT because LFT2 and PBFT have two phases and three phases on average, respectively.

\smallskip\noindent\textbf{Decentralization. } 
From the decentralization point of view, PBFT is the best because only one of the three phases is involved with a leader. 
In PBFT, only pre-prepare phase is a leader-dependent phase, but the other phases, prepare and commit phases, are leader-independent phases. 
Therefore, a proportion of leader-involved phases is 33.3\%, where a proportion of leader-involved phases indicates the ratio of the number of leader-involved phases to the number of phases in one round.
In Chained Hotstuff, however, every phase is involved with a leader. 
Every phase of Chained Hotstuff is symmetric, and a leader collects the vote messages, so a proportion of leader-involved phases is 100\%.
We can also consider that LFT2 is positioned between PBFT and Chained Hotstuff. 
In LFT2, the propose phase is the leader-dependent phase but vote phase is the leader-independent phase, which implies that a proportion of leader-involved phases is 50\%. 
Considering the above, we find out that PBFT has the best decentralization, LFT2 is the second, and Chained Hotstuff is the worst decentralization level.

\begin{table}[h]
\caption[Comparison among PBFT, LFT2 and Chained Hotstuff]{Comparison among the three consensus algorithms}
\vspace{2mm}
\centering
\begin{tabular}{|l|l|l|l|}
\hline
 & PBFT & LFT2 & Chained Hotstuff \\ \hline
\begin{tabular}[c]{@{}l@{}}Scalability 
(Average number of phases \\ per a committed block)\end{tabular} & Worst (3)& Second (2) & Best (1) \\ \hline
Network Bandwidth (Big-O notation)& Worst ($O(N^2)$) & Second ($O(N^2)$) & Best ($O(N)$)\\ \hline
\begin{tabular}[c]{@{}l@{}}Decentralization \\ (Proportion of leader-involved phases)\end{tabular}  & Best (33.3\%)    & Second (50\%) & Worst (100\%)    \\ \hline
\end{tabular}%
\end{table}

\section{Simulation}

In this chapter, we simulate the LFT2 consensus algorithm using LFT2 implementation~\cite{LFT2}. 
Our simulation measures $\gamma-function$ by varying two timeouts: $Propose Timeout$ and $Vote Timeout.$ 
In our simulation, 1) we investigated the relationship between timeout and $\gamma-function$ by varying the number of nodes, and 2) the relationship between timeout and $\gamma-function$ by varying the number of failure nodes. 
We set $Propose Timeout$ and $Vote Timeout$ to the same value, as ICON LOOP~\cite{ICON} set them, and simulate the timeout range from 0 to 4 seconds in 0.1 second increments.

\subsection{Methodology}

The LFT2 implementation provides a simulation tool with a system console, which can control the simulation environments. 
In the original implementation code, the network delay was set to a random value between 0 and 1 second, and the $Propose Timeout$ and $Vote Timeout$ was set to 2 seconds. 
However, in the real network, delay doesn't follow the random function, so we had to change the network delay model to a similar one to real. 
Because network delay changes dynamically for various reasons, we choose to collect real data rather than modeling network delay theoretically. 
We collected a data for network delay in a website~\cite{Bitcoin_Monitoring}, which provides some network information about Bitcoin. 
We use the `current block propagation delay distribution' of 3 February, 2020. 
The delay distribution represents the elapsed time between the first reception of an INV message announcing a new block and the subsequent reception from other peers in the Bitcoin network. 
The delay information contains some delay data that is over 4 seconds, but the most of data is less than 4 seconds. 
Thus, we ignore the data which is over 4 seconds. 
The probability distribution of the Bitcoin network delay is shown in Figure ~\ref{fig:delay_distribution}.

\begin{figure}[ht]
    \centerline{\includegraphics[width=\columnwidth]{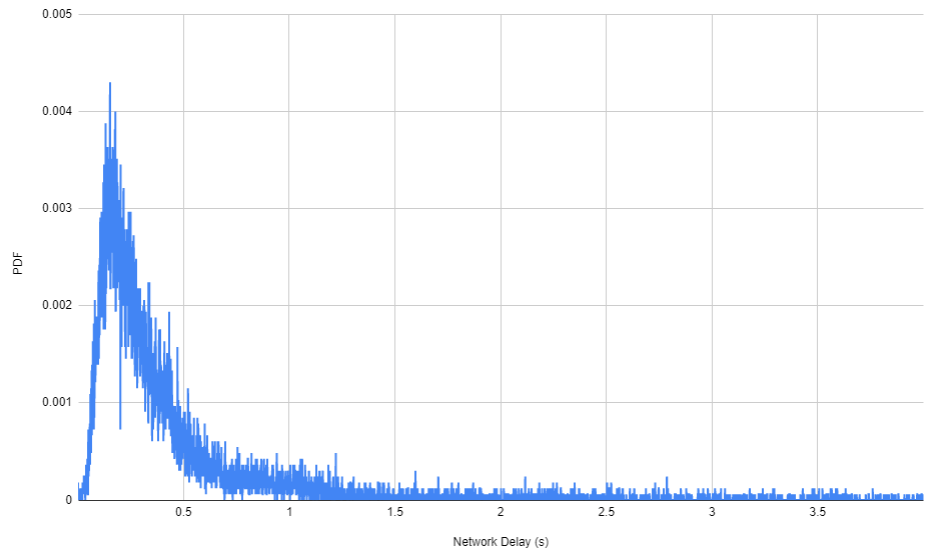}}
    \caption{Current block propagation delay distribution of the Bitcoin network on 3 February, 2020}
    \label{fig:delay_distribution}
\end{figure}

Using the simulation tool provided by ICON LOOP~\cite{LFT2}, we can measure $\gamma-function$ under various settings.
We do an experiment by changing two variables: the number of nodes in the LFT2 network and the number of failure nodes
In both cases, we measure the value of $\gamma-function$ with varying a timeout from 0 to 4 seconds in 0.1 second increments. 
Specifically, in the first, we measure $\gamma-function$ under four different settings of the total number of nodes, 4, 10, 50, and 100. 
In the second case, we measure $\gamma-function$ varying the number of failed nodes when total number of noes in LFT2 is 21. 
We calculated $\gamma-function$ by running the LFT2 simulator until more than 200 blocks are created.

\subsection{Result}

\begin{figure}[h]
    \centering
    \includegraphics[width=\columnwidth]{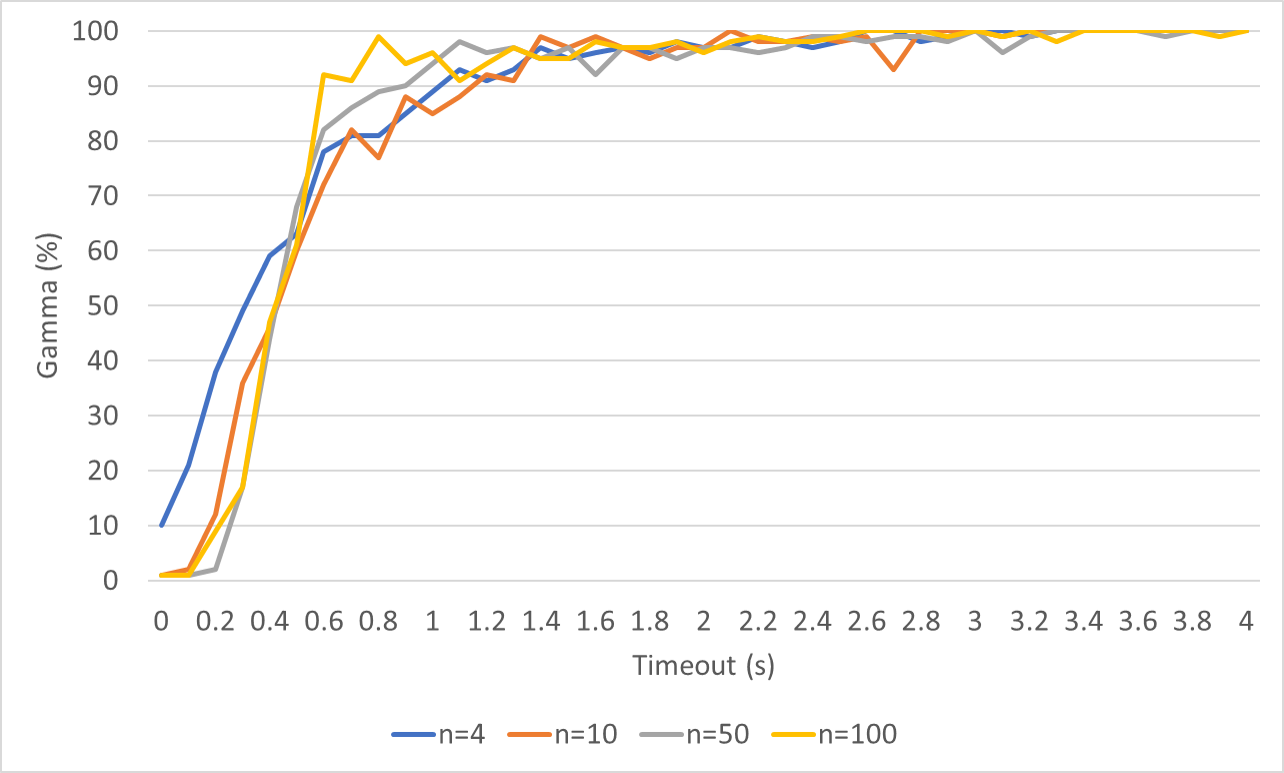}
    \caption{$\gamma$-function with varying timeout and the number of nodes}
    \label{fig:normal}
\end{figure}

\begin{figure}[h]
    \centering
    \includegraphics[width=\columnwidth]{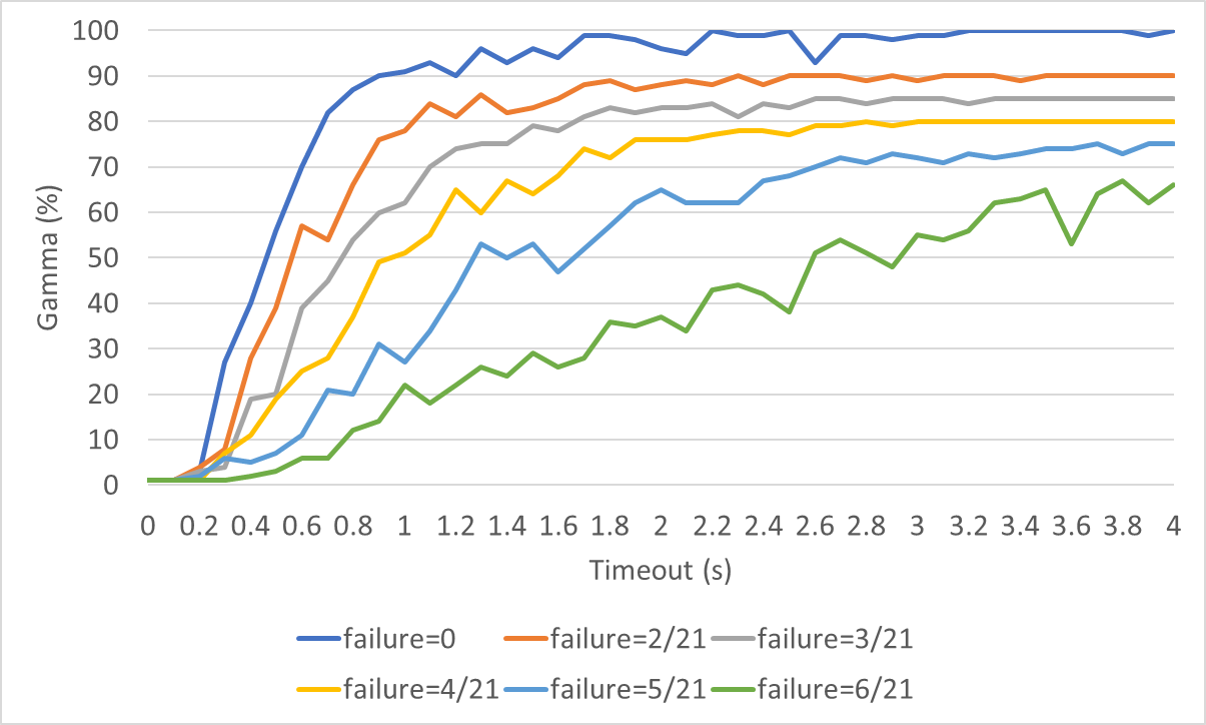}
    \caption{$\gamma$-function with varying timeout and the number of failure nodes}
    \label{fig:failure}
\end{figure}

\smallskip\noindent\textbf{The number of nodes. } 
Values of $\gamma-function$ with varying the number of nodes is represented in Figure~\ref{fig:normal}. 
In this figure, we can see that values of $\gamma-function$ are similar even when the number of nodes is different, which implies that $\gamma-function$ is rarely influenced by the number of nodes. 
Also, the value of $\gamma-function$ rises sharply between 0.2 and 0.6 seconds because most of network delay is concentrated in 0.1 to 0.3 seconds (refer to Figure~\ref{fig:delay_distribution}). 
This supports that our timeout condition in the liveness proof is reasonable when we set the value of timeout to 2 seconds. 
Finally, from Figure~\ref{fig:normal}, we observe that $\gamma-function$ converges to 100\% at timeout of greater than 2 seconds, which means that without any failure node, LFT2 commits all blocks properly.

\smallskip\noindent\textbf{Failure. } 
We also simulate LFT2 by varying the extent of failure, where the number of all nodes is set to 21. 
Figure~\ref{fig:failure} represents the value of $\gamma-function$ when the number of failed nodes is from zero to six out of 21 and timeout is from 0 to 4 seconds.
As shown in the figure, one can see $\gamma-function$ converges to $1-failure$ as timeout increases, where $failure$ indicates a fraction of failure nodes to all nodes.
This conforms with Lemma~\ref{lem:gamma}.
When the value of $failure$ is less than or equal to $\frac{4}{21},$ $\gamma-function$ converges within an error range of about 2\% at timeout of $2.7$ s.
Meanwhile, when five and six nodes are failed, timeout is required to be greater than 3.9 s and 5.3 s, respectively, for $\gamma-function$ to converge. 
Table~\ref{tab:converge} represents timeout when the value of $\gamma-function$ reaches to the value of convergence for the first time, under various settings of $failure$. 
In Table~\ref{tab:converge}, one can see that the greater the value of $failure$ is, the greater the timeout when $\gamma-function$ reaches to the value of convergence for the first time.

\begin{table}[h]
\caption{Comparison among the three consensus algorithms}
\vspace*{2mm}
\centering
\begin{tabular}{|p{2cm}|p{1cm}|p{1cm}| p{1cm}|p{1cm}|p{1cm}|p{1cm}|}
\hline
Failure & 0 & 2/21 & 3/21 & 4/21 & 5/21 & 6/21 \\ \hline
Timeout (s) & 2.2 & 2.3 & 2.6 & 2.8 & 3.7 & 4.9\\ \hline
\end{tabular}
\label{tab:converge}
\end{table}

\section{Conclusion} 

Many blockchain consensus algorithms have been developed, and one of most important properties of the consensus algorithm is safety and liveness. 
In this paper, we analyze LFT2 used in ICON. 
To do this, we formalize the protocol and analyze safety and liveness of LFT2.
It requires a certain network assumption to prove liveness. 
Therefore, to decide if this assumption is reasonable, we simulate LFT2 and measure liveness quality by using our metric, $\gamma-function.$ 
This shows that when we set timeout to a sufficiently large value (about 4 seconds), a high level of liveness can be guaranteed.

\bibliographystyle{splncs04}
\bibliography{references}

\begin{thebibliography}{10}
\providecommand{\url}[1]{\texttt{#1}}
\providecommand{\urlprefix}{URL }
\providecommand{\doi}[1]{https://doi.org/#1}

\bibitem{Bitcoin_Monitoring}
Bitcoin monitoring, \url{https://dsn.tm.kit.edu/bitcoin/}

\bibitem{ICON}
Icon loop web homepage, \url{https://www.iconloop.com/}

\bibitem{castro1999correctness}
Castro, M., Liskov, B., et~al.: A correctness proof for a practical
  byzantine-fault-tolerant replication algorithm. Tech. rep., Technical Memo
  MIT/LCS/TM-590, MIT Laboratory for Computer Science (1999)

\bibitem{castro1999practical}
Castro, M., Liskov, B., et~al.: Practical byzantine fault tolerance. In: OSDI.
  vol.~99, pp. 173--186 (1999)

\bibitem{fischer1982impossibility}
Fischer, M.J., Lynch, N.A., Paterson, M.S.: Impossibility of distributed
  consensus with one faulty process. Tech. rep., Massachusetts Inst of Tech
  Cambridge lab for Computer Science (1982)

\bibitem{LFT2}
ICONLOOP: Lft2, \url{https://github.com/icon-project/LFT2}

\bibitem{lft2_white}
ICONLOOP: Lft2 whitepaper,
  \url{https://github.com/icon-project/LFT2/blob/master/Whitepaper%20-%20LFT2%20(ENG).pdf}

\bibitem{Nakamoto_bitcoin:a}
Nakamoto, S.: Bitcoin: A peer-to-peer electronic cash system,
  http://bitcoin.org/bitcoin.pdf (2008)

\bibitem{hotstuffimage}
Ontology: HotStuff: the consensus protocol behind Facebook’s LibraBFT,
  \url{https://www.tokenbank.co.kr/news/detail/36188/}

\bibitem{yin2018hotstuff}
Yin, M., Malkhi, D., Reiter, M.K., Gueta, G.G., Abraham, I.: Hotstuff: Bft
  consensus in the lens of blockchain. arXiv preprint arXiv:1803.05069  (2018)

\end{thebibliography}

\end{document}